\documentclass{IEEEtran}

\usepackage[noadjust]{cite}
\usepackage{CJKutf8, amsmath, amssymb, orcidlink, dsfont}
\usepackage[T1]{fontenc}

\newtheorem{lemma}{Lemma}
\newtheorem{theorem}{Theorem}
\newtheorem{conjecture}{Conjecture}
\newtheorem{corollary}{Corollary}

\newtheorem{definition}{Definition}
\newtheorem{assumption}{Assumption}
\newtheorem{question}{Question}
\newtheorem{example}{Example}

\newtheorem{remark}{Remark}

\DeclareMathOperator{\tr}{tr}
\DeclareMathOperator{\Span}{span}
\DeclareMathOperator*{\E}{\mathbb E}

\begin{document}

\title{Deviation From Maximal Entanglement for Mid-Spectrum Eigenstates of Local Hamiltonians}

\begin{CJK}{UTF8}{gbsn}

\author{Yichen Huang (黄溢辰)\orcidlink{0000-0002-8496-9251}
\thanks{This work, dedicated to the Chinese New Year of the Tiger, was supported by NSF grant PHY-1818914 and the U.S. Department of Energy, Office of Science, National Quantum Information Science Research Centers, Quantum Systems Accelerator.}
\thanks{The author was with the Center for Theoretical Physics, Massachusetts Institute of Technology, Cambridge, MA 02139 USA, and also with the Department of Physics, Harvard University, Cambridge, MA 02138 USA (e-mail: huangtbcmh@gmail.com).}
}

\maketitle

\end{CJK}

\begin{abstract}

In a spin chain governed by a local Hamiltonian, we consider a microcanonical ensemble in the middle of the energy spectrum and a contiguous subsystem whose length is a constant fraction of the system size. We prove that if the bandwidth of the ensemble is greater than a certain constant, then the average entanglement entropy (between the subsystem and the rest of the system) of eigenstates in the ensemble deviates from the maximum entropy by at least a positive constant. This result highlights the difference between the entanglement entropy of mid-spectrum eigenstates of (chaotic) local Hamiltonians and that of random states. We also prove that the former deviates from the thermodynamic entropy at the same energy by at least a positive constant.

\end{abstract}


Preprint number: MIT-CTP/5396

\section{Introduction}

Almost all quantum many-body systems are chaotic.\footnote{This statement is not mathematically precise because ``chaotic'' is not defined. In this paper we do not attempt to define it, for there is no clear-cut definition of quantum chaos.} Therefore, quantum chaos is an important research topic in condensed matter and statistical physics. It is also of fundamental interest in high energy physics, since black holes are believed to be maximally chaotic in the sense of being the fastest scramblers in nature \cite{SS08}.

Although generic chaotic systems are not exactly solvable, an analytical understanding of quantum chaos may be gained from the following assumption.

\begin{assumption}
The properties of a quantum chaotic system are described by a random state (subject to macroscopic constraints if any).
\end{assumption}

\begin{remark} \label{a:1}
This assumption underlies the typicality argument \cite{Llo88, GLTZ06, PSW06}, which explains the emergence of the canonical ensemble in isolated quantum many-body systems.
\end{remark}

Using techniques from probability theory, one can derive some consequences of the random-state description in Assumption \ref{a:1}. One might expect that such consequences are universal properties of quantum chaotic systems. This approach is elegant and independent of the microscopic details of the system under consideration. In various specific models, certain consequences of Assumption \ref{a:1} have been numerically observed \cite{KH13, YCHM15, VR17, Hua19NPB} or even rigorously proved \cite{YGC23, HH23} to be at least approximately valid. The empirical success of the random-state description (Assumption \ref{a:1}) may be surprising especially when the chaotic system under consideration is deterministic (not random).

Physical systems usually have some structure. In a quantum spin system on a lattice, interactions among spins may be local in the sense that the interaction range of each term in the Hamiltonian is upper bounded by a constant. Since locality is not captured by Assumption \ref{a:1}, it is important to understand whether and to what extent the properties of chaotic systems with local interactions deviate from the consequences of the random-state description. This is a technically challenging task---it is difficult to verify or refute Assumption \ref{a:1} for a particular chaotic system. Whether the deviation exists may depend on the property under consideration.

Entanglement, a concept of quantum information theory, has been widely used in condensed matter and statistical physics to provide insights beyond those obtained via ``conventional'' quantities. A large body of literature exists on ground-state entanglement \cite{HLW94, VLRK03, KP06, LW06, Has07} and entanglement dynamics \cite{CC05, ZPP08, BPM12, KH13} in various systems. The scaling of entanglement \cite{ECP10} reflects the efficiency of tensor network simulations of quantum many-body systems \cite{Vid03, VC06, SWVC08, Osb12, GHLS15}.

It is important and of high current interest to understand the entanglement of excited eigenstates of chaotic local Hamiltonians. To this end, heuristic descriptions of universal entanglement scaling behavior have been developed \cite{Deu10, SPR12, DLS13, YCHM15, DKPR16, VR17, DLL18, GG18, LG19, Hua19NPB, Hua21NPB, MTA+21, HMK22}. In particular, Assumption \ref{a:1} raises the following question.

\begin{question} \label{q:1}
Does the difference between the entanglement entropy of mid-spectrum eigenstates of chaotic local Hamiltonians and that of random states vanish in the thermodynamic limit?
\end{question}

The thermodynamic limit is the limit where the system size (number of spins in a spin system) goes to infinity.

Assume without loss of generality that the Hamiltonian under consideration is traceless so that the mean energy (average of all eigenvalues) of the Hamiltonian is zero.

\begin{definition} [mid-spectrum eigenstate] \label{d:mid}
For a traceless Hamiltonian, an eigenstate with energy $E$ is a mid-spectrum eigenstate if the ratio of the number of eigenvalues in the interval $[-|E|,|E|]$ to the Hilbert space dimension vanishes in the thermodynamic limit. Otherwise, it is a non-mid-spectrum eigenstate.
\end{definition}

Question \ref{q:1} concerns only mid-spectrum eigenstates because the energy of a random state approaches zero in the thermodynamic limit. For a non-mid-spectrum eigenstate of any (not necessarily chaotic) local Hamiltonian and for contiguous subsystems whose size is a constant fraction of the system size, Lemma \ref{l} (see also Lemma \ref{l:thermo}) below implies that the entanglement entropy deviates from the maximum entropy by at least a positive constant. Thus, the entanglement entropy of a non-mid-spectrum eigenstate is different from that of a random state. Question \ref{q:1} concerns mid-spectrum eigenstates and is not answered by Lemma \ref{l} or \ref{l:thermo}. Theorem 2 in Ref.~\cite{Hua19NPB} considers the average entanglement entropy of all eigenstates of a local Hamiltonian and gives a lower bound on its deviation from the maximum entropy. It does not answer Question \ref{q:1} because only a vanishing fraction of eigenstates are mid-spectrum eigenstates (Definition \ref{d:mid}).

Question \ref{q:1} has been studied. The difference between the entanglement entropy of mid-spectrum eigenstates of chaotic local Hamiltonians and that of random states was observed numerically in Refs.~\cite{YCHM15, GG18}. However, to my knowledge, it was highlighted only in recent works \cite{HMK22, KSVR23, RJK24}.\footnote{After this paper was posted on arXiv, Refs.~\cite{KSVR23, RJK24} appeared.} These numerical studies are limited to relatively small system sizes and suffer from significant finite-size effects. The numerical values of the difference in all these studies are quite small ($\lesssim0.1$). Strictly speaking, whether they vanish in the thermodynamic limit cannot be concluded from the numerical results. Heuristic explanations for the numerically observed differences were given in Refs.~\cite{Hua21NPB, HMK22}. It was even conjectured in Eqs.~(14), (15) of Ref.~\cite{Hua21NPB} that in the thermodynamic limit, the difference in Question \ref{q:1} approaches a universal (model-independent) function of the ratio of the subsystem size to the system size.

This paper focuses on one-dimensional quantum systems and rigorously answers Question \ref{q:1} in the negative. In a spin chain governed by a (not necessarily chaotic) local Hamiltonian, we consider a microcanonical ensemble in the middle of the energy spectrum and a contiguous subsystem whose length is a constant fraction of the system size. We prove that \emph{if the bandwidth of the ensemble is greater than a certain constant, then the average entanglement entropy (between the subsystem and the rest of the system) of eigenstates in the ensemble deviates from the maximum entropy\footnote{The maximum entropy is the logarithm of the Hilbert space dimension of the smaller of the subsystem and the rest of the system.} by at least a positive constant.} This result implies and highlights the difference between the entanglement of mid-spectrum eigenstates of (chaotic) local Hamiltonians and that of random states. The difference is due to the locality of the Hamiltonian. The locality prevents mid-spectrum eigenstates from behaving like completely random states. We also prove that the entanglement entropy of mid-spectrum eigenstates deviates from the thermodynamic entropy at the same energy by at least a positive constant. Therefore, the italicized statement above cannot be understood heuristically from eigenstate thermalization.

The rest of this paper is organized as follows. Section \ref{spre} provides basic definitions and technical background. Section \ref{sres} presents the main result, whose relationship with the thermodynamic entropy is discussed in Section \ref{seth}. The main text of this paper should be easy to read, for most of the technical details are deferred to Appendices \ref{app}, \ref{s:l}.

\section{Preliminaries} \label{spre}

We use standard asymptotic notation. Let $f,g:\mathbb R^+\to\mathbb R^+$ be two functions. One writes $f(x)=O(g(x))$ if and only if there exist constants $M,x_0>0$ such that $f(x)\le Mg(x)$ for all $x>x_0$; $f(x)=\Omega(g(x))$ if and only if there exist constants $M,x_0>0$ such that $f(x)\ge Mg(x)$ for all $x>x_0$; $f(x)=\Theta(g(x))$ if and only if there exist constants $M_1,M_2,x_0>0$ such that $M_1g(x)\le f(x)\le M_2g(x)$ for all $x>x_0$; $f(x)=o(g(x))$ if and only if for any constant $M>0$ there exists a constant $x_0>0$ such that $f(x)<Mg(x)$ for all $x>x_0$.

Consider a chain of $N$ spins, each of which has local dimension $d=\Theta(1)$. The system is governed by a local Hamiltonian
\begin{equation} \label{eq:LH}
H=\sum_{i=1}^NH_i,
\end{equation}
where $H_i$ is the nearest-neighbor interaction between spins at sites $i$ and $(i\bmod N)+1$. For concreteness, we use periodic boundary conditions, but all the results below also hold for other boundary conditions. Assume without loss of generality that $\tr H_i=0$ (traceless) for all $i$ so that the mean energy of $H$ is $\tr H/d^N=0$. We also assume that the expansion of $H_i$ in the generalized Pauli basis does not contain any terms acting only on the spin at site $(i\bmod N)+1$. This is again without loss of generality since such terms can be included in $H_{(i\bmod N)+1}$. These two assumptions imply that
\begin{equation} \label{eq:neq}
\tr(H_iH_{i'})=0,\quad\forall i\neq i'.
\end{equation}
We do not assume translational invariance. $\|H_i\|$ may be site dependent but should be $\Theta(1)$ for all $i$.\footnote{The assumption that $\|H_i\|=\Theta(1)$ for all $i$ can be relaxed. The proof of Theorem \ref{thm:main} remains valid if $\|H_i\|=O(1)$ for all $i$ and if condition (\ref{eq:var}) holds.}

\begin{lemma} [\cite{HBZ19, Hua22AP, Hua22ATMP}] \label{l:vars}
\begin{equation} \label{eq:vars}
s^2:=\tr(H^2)/d^N=\Theta(N).
\end{equation}
\end{lemma}

\begin{remark}
$s^2$ is the variance of all eigenvalues of $H$.
\end{remark}

\begin{IEEEproof}
Using Eq.~(\ref{eq:neq}) and since $\|H_i\|=\Theta(1)$ for all $i$,
\begin{multline}
\frac{\tr(H^2)}{d^N}=\sum_{i,i'=1}^N\frac{\tr(H_iH_{i'})}{d^N}=\sum_{i=1}^N\frac{\tr(H_i^2)}{d^N}=\sum_{i=1}^N\Theta(1)\\
=\Theta(N).
\end{multline}
\end{IEEEproof}

Let $A$ be a contiguous subsystem of $L$ spins and $\bar A$ be the complement of $A$ (rest of the system). Assume without loss of generality that $L\le N/2$. Let $d_A:=d^L$ and $d_{\bar A}:=d^{N-L}$ be the Hilbert space dimensions of subsystems $A$ and $\bar A$, respectively. 

\begin{definition} [entanglement entropy]
The entanglement entropy of a pure state $|\psi\rangle$ is defined as the von Neumann entropy
\begin{equation}
S(\psi_A)=-\tr(\psi_A\ln\psi_A)
\end{equation}
of the reduced density matrix $\psi_A:=\tr_{\bar A}|\psi\rangle\langle\psi|$.
\end{definition}

We call $\max_{|\psi\rangle}S(\psi_A)=\ln d_A=L\ln d$ the maximum entropy.

\paragraph{Entanglement of random states}We briefly review the entanglement of random states.

\begin{theorem} [conjectured and partially proved by Page \cite{Pag93}; proved in Refs.~\cite{FK94, San95, Sen96}] \label{t:page}
For a pure state $|\psi\rangle$ chosen uniformly at random with respect to the Haar measure,
\begin{equation} \label{page}
\E_{|\psi\rangle}S(\psi_A)=\sum_{k=d_{\bar A}+1}^{d_Ad_{\bar A}}\frac1k-\frac{d_A-1}{2d_{\bar A}}=\ln d_A-\frac{d_A}{2d_{\bar A}}+\frac{O(1)}{d_Ad_{\bar A}}.
\end{equation}
\end{theorem}

The second step of Eq.~(\ref{page}) uses the formula
\begin{equation}
\sum_{k=1}^{n}\frac1k=\ln n+\gamma+\frac1{2n}+O(1/n^2)
\end{equation}
for $n=d_{\bar A}$ and $n=d_Ad_{\bar A}$, where $\gamma\approx0.577216$ is the Euler-Mascheroni constant.

The distribution of $S(\psi_A)$ is highly concentrated around the mean $\E_{|\psi\rangle}S(\psi_A)$ \cite{HLW06}. This can also be seen from the exact formula \cite{VPO16, Wei17} for the variance of $S(\psi_A)$.

Suppose $f:=L/N$ is a fixed constant such that $0<f\le1/2$. Theorem \ref{t:page} implies that in the thermodynamic limit $N\to\infty$, 
\begin{equation}
    \E_{|\psi\rangle}S(\psi_A)=L\ln d-d^{(2f-1)N}/2+O(d^{-N}).
\end{equation}
Thus, for $0<f<1/2$, the difference between the entanglement entropy of random states and the maximum entropy is exponentially small $e^{-\Omega(N)}$. For $f=1/2$ (equal bipartition), the difference is exponentially close to $1/2$.

\paragraph{Entanglement of non-mid-spectrum eigenstates}Recall that we assume that the Hamiltonian $H$ (\ref{eq:LH}) is traceless. If the energy of a (possibly mixed) state $\rho$ is significantly away from zero, then the subsystem entropy deviates significantly from the maximum entropy in the following sense. Let $\rho_A:=\tr_{\bar A}\rho$ be the reduced density matrix of subsystem $A$. Let $\E_{|A|=L}$ denote averaging over all contiguous subsystems of length $L$. There are $N$ such subsystems.

\begin{lemma} \label{l}
For $L>1$,\footnote{On the left-hand side of Eq.~(\ref{eq:up}), we average over \emph{all} contiguous subsystems of length $L$. This is not strictly necessary for deriving the bound on the right-hand side. For example, if $L=N/2$ (which implies that $N$ is even) and if $|\tr(\rho H)|$ is greater than a certain constant, then the derivation remains valid upon replacing $\E_{|A|=L}$ by averaging over only two contiguous subsystems $A$ and $\bar A$. If $\rho$ is a pure state, then
\begin{equation}
S(\rho_A)=\frac{S(\rho_A)+S(\rho_{\bar A})}2=\frac N2\ln d-\Omega(\tr^2(\rho H)/N).
\end{equation}}
\begin{equation} \label{eq:up}
\E_{|A|=L}S(\rho_A)=L\ln d-\Omega(L\tr^2(\rho H)/N^2).
\end{equation}
\end{lemma}

\begin{IEEEproof}
For $d=2$, the lemma is proved in Ref.~\cite{Hua19NPB}.\footnote{In the proof of Ref.~\cite{Hua19NPB}, $\rho$ is an eigenstate of $H$. However, the proof does not use this property of $\rho$. Thus, without modification it applies to any (possibly mixed) state.} It is straightforward to extend the proof to any integer $d>2$.
\end{IEEEproof}

It is straightforward to extend Lemma \ref{l} to two and higher spatial dimensions.

Let $\{|\Psi_j\rangle\}_{j=1}^{d^N}$ be a complete set of eigenstates of $H$ with corresponding energies $\{E_j\}$. Since $H$ is a local Hamiltonian, the distribution of $E_j$'s (density of states) is approximately normal \cite{KLW15, RCA24, BC15} with mean $\tr H/d^N=0$ and variance $s^2$ (\ref{eq:vars}). Hence, $|\Psi_j\rangle$ is a mid-spectrum or non-mid-spectrum eigenstate (Definition \ref{d:mid}) if $|E_j|=o(\sqrt N)$ or $|E_j|=\Omega(\sqrt N)$, respectively. Let $\Psi_{j,A}:=\tr_{\bar A}|\Psi_j\rangle\langle\Psi_j|$ be the reduced density matrix of subsystem $A$. For any non-mid-spectrum eigenstate and $L=\Omega(N)$, Lemma \ref{l} implies that
\begin{equation}
\E_{|A|=L}S(\Psi_{j,A})=L\ln d-\Omega(1).
\end{equation}
By contrast, Lemma \ref{l} cannot determine whether the entanglement entropy of a mid-spectrum eigenstate deviates from the maximum entropy by at least a positive constant.

\section{Results} \label{sres}

Let
\begin{equation}
J=\{j:E-\Delta\le E_j\le E+\Delta\}
\end{equation}
be a microcanonical ensemble of bandwidth $2\Delta$. The main result of this paper, proved in Appendix \ref{app}, is the following theorem.

\begin{theorem} \label{thm:main}
Suppose $L=\Omega(N)$. There is a constant $C>0$ such that for any $|E|=o(\sqrt N)$ and any $\Delta>C$,
\begin{equation} \label{eq:main}
\frac1{|J|}\sum_{j\in J}S(\Psi_{j,A})=L\ln d-\Omega(1).
\end{equation}
\end{theorem}

Theorem \ref{thm:main} is particularly interesting when $C<\Delta=o(\sqrt N)$. In this case, $|E_j|=o(\sqrt N)$ so that $|\Psi_j\rangle$ is a mid-spectrum eigenstate for all $j\in J$.

Comparing Theorem \ref{t:page} and Theorem \ref{thm:main} for $C<\Delta=o(\sqrt N)$, the entanglement entropy of mid-spectrum eigenstates of any one-dimensional local Hamiltonian is provably different from that of random states. The difference is $\Omega(1)$ if $N/2-C'>L=\Omega(N)$ for a certain constant $C'>0$. We conjecture that the difference is also $\Omega(1)$ if $N/2-C'\le L\le N/2$. For $L=N/2$, Refs.~\cite{YCHM15, GG18, Hua21NPB, HMK22, KSVR23, RJK24} provide numerical evidence for this conjecture.

\paragraph{Superposition of eigenstates}Let $\{|\psi_j\rangle\}_{j=1}^{|J|}$ be an arbitrary orthonormal basis of the subspace
\begin{equation}
\mathcal H_J:=\Span\{|\Psi_j\rangle:j\in J\}.
\end{equation}
Let $\psi_{j,A}:=\tr_{\bar A}|\psi_j\rangle\langle\psi_j|$ be the reduced density matrix of subsystem $A$.

\begin{theorem} \label{t:a}
Suppose $L=\Omega(N)$. There is a constant $C>0$ such that for any $|E|=o(\sqrt N)$ and for any $\Delta$ such that $C<\Delta=O(\sqrt N)$,
\begin{equation} \label{eq:a}
\frac1{|J|}\sum_{j=1}^{|J|}S(\psi_{j,A})=L\ln d-\Omega(1).
\end{equation}
\end{theorem}

\begin{remark}
Corollary \ref{c:r} in Appendix \ref{app} extends Theorems \ref{thm:main} and \ref{t:a} to the R\'enyi entropy of $\Psi_{j,A}$ and $\psi_{j,A}$, respectively.
\end{remark}

\begin{corollary} \label{cor}
Let $|\psi\rangle$ be a pure state chosen uniformly at random with respect to the Haar measure from $\mathcal H_J$. Suppose $L=\Omega(N)$. There is a constant $C>0$ such that for any $|E|=o(\sqrt N)$ and for any $\Delta$ such that $C<\Delta=O(\sqrt N)$,
\begin{equation}
\Pr_{|\psi\rangle\in\mathcal H_J}\big(S(\psi_A)=L\ln d-\Omega(1)\big)=1-e^{-\Omega(d^N\Delta/N^{5/2})}.
\end{equation}
\end{corollary}

We split the Hamiltonian (\ref{eq:LH}) into three parts: $H=H_A+H_{\bar A}+H_\partial$, where $H_{A(\bar A)}$ contains terms acting only on subsystem $A(\bar A)$, and $H_\partial$ contains boundary terms. For example, if subsystem $A$ consists of spins at sites $1,2,\ldots,L$, then
\begin{equation} \label{eq:hdef}
H_A=\sum_{i=1}^{L-1}H_i,\quad H_{\bar A}=\sum_{i=L+1}^{N-1}H_i,\quad H_\partial=H_L+H_N.
\end{equation}
Note that $H_A$ and $H_{\bar A}$ are Hermitian matrices of order $d_A$ and $d_{\bar A}$, respectively. Let
\begin{equation} \label{eq:def}
s_A^2:=\tr(H_A^2)/d_A,\quad s_{\bar A}^2:=\tr(H_{\bar A}^2)/d_{\bar A}.
\end{equation}

\begin{conjecture} \label{cuni}
Let $|\psi\rangle$ be a pure state chosen uniformly at random with respect to the Haar measure from $\mathcal H_J$. Let $f:=L/N$ be a fixed constant and $\delta$ denote the Kronecker delta. If
\begin{equation} \label{eq:cond}
\lim_{N\to\infty}s_A^2/s^2=f,
\end{equation}
then for $d=2$, any $|E|=o(\sqrt N)$, and any $\Delta=o(\sqrt N)$,
\begin{equation}
\lim_{N\to\infty}\left(L\ln2-\E_{|\psi\rangle\in\mathcal H_J}S(\psi_A)\right)=\frac{\delta_{f,1/2}-f-\ln(1-f)}2.
\end{equation}
\end{conjecture}

The proof of Lemma \ref{l:vars} implies that Eq.~(\ref{eq:cond}) holds for not only any translation-invariant but also many disordered systems. For example, Eq.~(\ref{eq:cond}) holds with overwhelming probability if each $H_i$ is a random Hermitian matrix of order $d^2$ sampled independently from the Gaussian unitary ensemble.

Conjecture \ref{cuni} can be heuristically justified in the same way as Eqs.~(14), (15) of Ref.~\cite{Hua21NPB}. Condition (\ref{eq:cond}) is not but should have been included in Ref.~\cite{Hua21NPB} because it is necessary for the heuristic justification there.

\section{Eigenstate thermalization} \label{seth}

The eigenstate thermalization hypothesis (ETH) states that for expectation values of local observables, a single eigenstate resembles a thermal state with the same energy \cite{Deu91, Sre94, RDO08}. An entropic variant of the ETH states that the entanglement entropy of an eigenstate between a contiguous subsystem (smaller than half the system size) and the rest of the system is approximately equal to the thermodynamic entropy of the subsystem at the same energy \cite{Deu10, SPR12, DLS13, DKPR16, GG18}.

For any local Hamiltonian (regardless of whether it is chaotic), the subsystem entropy of a (possibly mixed) state $\rho$ is, in the following sense, upper bounded by the thermodynamic entropy at the same energy. Let
\begin{equation}
\sigma_A(\beta):=e^{-\beta H_A}/\tr e^{-\beta H_A}
\end{equation}
be the thermal state of subsystem $A$ at inverse temperature $\beta$. Recall the definition of $\E_{|A|=L}$ in the paragraph preceding Lemma \ref{l}. Let
\begin{equation}
\mathcal E(\beta):=\E_{|A|=L}\tr(\sigma_A(\beta)H_A).
\end{equation}
It is easy to see that $\mathcal E(0)=0$ and that $\mathcal E$ is strictly monotonically decreasing. The effective inverse temperature of $\rho$ is determined by solving the equation
\begin{equation} \label{eq:eng}
\mathcal E(\beta)=(L-1)\tr(\rho H)/N.
\end{equation}
Note that ``$L-1$'' on the right-hand side is the number of terms in $H_A$. Let
\begin{equation}
\mathcal S(\beta):=\E_{|A|=L}S(\sigma_A(\beta)).
\end{equation}
It is easy to see that $\mathcal S(0)=L\ln d$ and that $\mathcal S$ is strictly monotonically increasing (decreasing) for negative (positive) $\beta$.

\begin{lemma} \label{l:thermo}
\begin{equation} \label{eq:thermo}
\E_{|A|=L}S(\rho_A)\le\mathcal S(\beta),
\end{equation}
where $\beta$ is the solution of Eq.~(\ref{eq:eng}).
\end{lemma}

\begin{IEEEproof}
It follows from Lemma 11 in Ref.~\cite{Hua21ISIT} or in Ref.~\cite{Hua22TIT} and the observation that
\begin{multline}
\E_{|A|=L}\tr(\rho_AH_A)=\tr\left(\rho\E_{|A|=L}H_A\otimes\mathds1_{\bar A}\right)\\
=(L-1)\tr(\rho H)/N,
\end{multline}
where $\mathds1_{\bar A}$ is the identity operator on subsystem $\bar A$.
\end{IEEEproof}

It is straightforward to extend Lemma \ref{l:thermo} to two and higher spatial dimensions.

\begin{question} \label{q:2}
For eigenstates of chaotic local Hamiltonians, does the difference between the left-hand side (entanglement entropy) and the right-hand side (thermodynamic entropy) of (\ref{eq:thermo}) vanish in the thermodynamic limit?
\end{question}

Suppose $L=\Omega(N)$. We answer Question \ref{q:2} for mid-spectrum eigenstates, i.e., $|\Psi_j\rangle$ with $|E_j|=o(\sqrt N)$. In this case, the right-hand side of (\ref{eq:thermo}) can be rigorously calculated: Lemma \ref{l:et} in Appendix \ref{s:l} implies that the solution of Eq.~(\ref{eq:eng}) is
\begin{equation}
\beta=-\Theta(E_j/N).
\end{equation}
Then, Lemma \ref{l:ent} in Appendix \ref{s:l} implies that
\begin{equation} \label{eq:thent}
\mathcal S(\beta)=L\ln d-\Theta(LE_j^2/N^2)=L\ln d-o(1).
\end{equation}
Thus, the bound in Lemma \ref{l:thermo} is the same as that in Lemma \ref{l}. Comparing Eqs.~(\ref{eq:main}) and (\ref{eq:thent}), the answer to Question \ref{q:2} for mid-spectrum eigenstates is no.

\appendices

\section{Proofs for Section \ref{sres}} \label{app}

\begin{IEEEproof} [Proof of Theorem \ref{thm:main}]
Since the distribution of $E_j$'s is approximately normal \cite{KLW15, RCA24, BC15} with mean zero and variance $s^2$ (\ref{eq:vars}),
\begin{align}
&1\ge\big|\{j:|E_j-E|\le\sqrt N\}\big|/|J|\nonumber\\
&\ge\big|\{j:|E_j-E|\le\sqrt N\}\big|/d^N=\Omega(1),\quad\forall\Delta>\sqrt N.
\end{align}
It suffices to prove Theorem \ref{thm:main} with the additional constraint that $\Delta\le\sqrt N$. This is a special case of Theorem \ref{t:a}.
\end{IEEEproof}

Before proving Theorem \ref{t:a}, it is instructive to verify Eq.~(\ref{eq:a}) in a simple example.

\begin{example}
Consider a chain of $N$ qubits (spin-$1/2$'s) governed by the Hamiltonian
\begin{equation} \label{eq:dh}
H=\sum_{i=1}^N\sigma_i^z,
\end{equation}
where $\sigma_i^z$ is the Pauli $z$ matrix at site $i$. Suppose $N$ is even. Let $E=0$ and $\Delta=1$ so that $|J|=\binom N{N/2}$. Let $\Pi=\Pi^2$ be the projection onto the subspace $\mathcal H_J=\ker H$ so that
\begin{multline}
\rho_A:=\frac1{|J|}\sum_{j=1}^{|J|}\psi_{j,A}=\frac1{|J|}\tr_{\bar A}\Pi\\
={\binom N{N/2}}^{-1}\bigoplus_{k=0}^L\binom{N-L}{N/2-k}I_{\binom Lk}
\end{multline}
in the computational basis, where $I_n$ is the identity matrix of order $n$. For $N/2\ge L=\Omega(N)$,
\begin{multline}
\frac1{|J|}\sum_{j=1}^{|J|}S(\psi_{j,A})\le S(\rho_A)=\sum_{k=0}^L\frac{\binom Lk\binom{N-L}{N/2-k}}{\binom N{N/2}}\ln\frac{\binom N{N/2}}{\binom{N-L}{N/2-k}}\\
=L\ln2-\Theta(1).
\end{multline}
\end{example}

(\ref{eq:dh}) is a very special Hamiltonian. For general $H$, $\rho_A$ cannot be calculated analytically. However, we can still prove that many diagonal elements of $\rho_A$ (in a particular basis) deviate significantly from $1/d_A$. Thus, $S(\rho_A)$ deviates significantly from the maximum entropy.

\begin{IEEEproof} [Proof of Theorem \ref{t:a}]
Recall the definitions of $s_A^2,s_{\bar A}^2$ (\ref{eq:def}). Similar to Lemma \ref{l:vars},
\begin{equation} \label{eq:var}
s_A^2=\Theta(L),\quad s_{\bar A}^2=\Theta(N-L).
\end{equation}
Let $\{|a_k\rangle\}_{k=1}^{d_A}$ and $\{|\bar a_l\rangle\}_{l=1}^{d_{\bar A}}$ be complete sets of eigenstates of $H_A$ and $H_{\bar A}$ with corresponding eigenvalues $\{\epsilon_k\}$ and $\{\varepsilon_l\}$, respectively. Assume without loss of generality that $E\le0$. Let
\begin{gather}
K:=\{k:C_1\sqrt L\le\epsilon_k<(C_1+1)\sqrt L\},\\
\Lambda:=\Delta+C_2,\quad Q_x:=\sum_{l:~|\varepsilon_l-x|\le\Lambda}|\bar a_l\rangle\langle\bar a_l|,\\
P:=\sum_{k=1}^{d_A}|a_k\rangle\langle a_k|\otimes(\mathds1_{\bar A}-Q_{E-\epsilon_k}),\\
P_K:=\sum_{k\in K}|a_k\rangle\langle a_k|\otimes(\mathds1_{\bar A}-Q_{E-\epsilon_k}),
\end{gather}
where $C_1,C_2>0$ are constants to be chosen later. The distributions of $E_j$'s, of $\epsilon_k$'s, and of $\varepsilon_l$'s are all approximately normal \cite{KLW15, RCA24, BC15}. Specifically, Lemma 1 in Ref.~\cite{RCA24} implies that
\begin{gather}
|J|/d^N\ge T_1-R_1,\quad m:=|K|/d_A\ge T_2-R_2,\\
\tr Q_{E-\epsilon_k}/d_{\bar A}\le T_3+R_3,\quad\forall k\in K,
\end{gather}
where
\begin{gather}
T_1=e^{-\frac{(E-\Delta)^2}{2s^2}}\frac\Delta{s}\sqrt{\frac2\pi},\quad T_2=e^{-\frac{(C_1+1)^2L}{2s_A^2}}\frac{\sqrt L}{s_A\sqrt{2\pi}},\\
T_3=e^{-\frac{(\max\{0,C_1\sqrt L-E-\Lambda\})^2}{2s_{\bar A}^2}}\frac\Lambda{s_{\bar A}}\sqrt{\frac2\pi},\\
R_1=c_1/\sqrt N,\quad R_2=c_1/\sqrt L,\quad R_3=c_1/\sqrt{N-L}
\end{gather}
for some constant $c_1>0$. Let
\begin{gather}
\rho:=\frac1{|J|}\sum_{j\in J}|\Psi_j\rangle\langle\Psi_j|=\frac1{|J|}\sum_{j=1}^{|J|}|\psi_j\rangle\langle\psi_j|,\\
\rho_A:=\tr_{\bar A}\rho=\frac1{|J|}\sum_{j=1}^{|J|}\psi_{j,A}.
\end{gather}
Recall the definition of $H_\partial$ (\ref{eq:hdef}). Since $\|H_\partial\|=O(1)$, Theorem 2.3 in Ref.~\cite{AKL16} implies that
\begin{equation}
\langle\Psi_j|P|\Psi_j\rangle=\||\Psi_j\rangle\langle\Psi_j|P\|^2\le c_2e^{c_3(\Delta-\Lambda)}=c_2e^{-c_3C_2}
\end{equation}
for all $j\in J$, where $c_2,c_3>0$ are constants. Hence,
\begin{equation}
\tr(\rho P)\le c_2e^{-c_3C_2}.
\end{equation}
For $x,y>0$, we write $x\ll y$ if $x/y$ is upper bounded by a sufficiently small positive constant. For a sufficiently large constant $C\gg1$,
\begin{equation}
R_1<R_3\ll T_1.
\end{equation}
Let $C_1,C_2$ be constants such that
\begin{equation}
\frac{\Delta^2}L+\frac NL\ln\left(2+\frac{C_2}C\right)\ll C_1^2\ll C_2.
\end{equation}
Since $R_2/T_2=O(1/\sqrt L)$ and $(C_1+1)^2\ll C_2$,
\begin{equation}
\frac{\tr(\rho P)}m\le\frac{c_2e^{-c_3C_2}}{T_2-R_2}=c_2s_Ae^{\frac{(C_1+1)^2L}{2s_A^2}-c_3C_2+o(1)}\sqrt{\frac{2\pi}L}\ll1.
\end{equation}
Since $0\le-E=o(\sqrt N)$ and $\Lambda/\Delta\le1+C_2/C$,
\begin{multline}
\frac{T_3}{T_1}=e^{\frac{(E-\Delta)^2}{2s^2}-\frac{(\max\{0,C_1\sqrt L-E-\Lambda\})^2}{2s_{\bar A}^2}}\frac{s\Lambda}{s_{\bar A}\Delta}\\
\le e^{\frac{\Delta^2}{2s^2}-\frac{(C_1\sqrt L-\Delta)^2}{2s_{\bar A}^2}+o(1)}\frac{s}{s_{\bar A}}\left(1+\frac{C_2}C\right)\ll1
\end{multline}
so that
\begin{equation}
\frac{d_A\tr Q_{E-\epsilon_k}}{|J|}\le\frac{T_3+R_3}{T_1-R_1}\ll1,\quad\forall k\in K.
\end{equation}
Let $p_k:=\langle a_k|\rho_A|a_k\rangle$ so that $\sum_{k=1}^{d_A}p_k=1$. Since $\rho\le I/|J|$,
\begin{align} \label{eq:cst}
\sum_{k\in K}p_k&=\sum_{k\in K}\langle a_k|\rho_A|a_k\rangle=\sum_{k\in K}\tr(\rho(|a_k\rangle\langle a_k|\otimes\mathds1_{\bar A}))\nonumber\\
&=\tr(\rho P_K)+\sum_{k\in K}\tr(\rho(|a_k\rangle\langle a_k|\otimes Q_{E-\epsilon_k}))\nonumber\\
&\le\tr(\rho P)+\sum_{k\in K}\frac{\tr Q_{E-\epsilon_k}}{|J|}\le\frac m4+\frac{|K|}{4d_A}=\frac m2.
\end{align}
Using the concavity of the von Neumann entropy and Eq.~(1.42) of Ref.~\cite{Weh78},
\begin{equation} \label{eq:concon}
\frac1{|J|}\sum_{j=1}^{|J|}S(\psi_{j,A})\le S(\rho_A)\le-\sum_{k=1}^{d_A}p_k\ln p_k.
\end{equation}
Maximizing the rightmost side subject to the constraint (\ref{eq:cst}) and since $m\ge T_2-R_2=\Theta(1)$,
\begin{multline}
-\sum_{k=1}^{d_A}p_k\ln p_k\le\frac{m\ln(2d_A)}2+\left(1-\frac m2\right)\ln\frac{d_A(1-m)}{1-m/2}\\
=L\ln d+\frac{m\ln2}2+\left(1-\frac m2\right)\ln\frac{1-m}{1-m/2}=L\ln d-\Theta(1).
\end{multline}
\end{IEEEproof}

\begin{definition} [R\'enyi entropy]
The R\'enyi entropy $S_\alpha$ of order $\alpha\in(0,1)\cup(1,\infty)$ of a state $\rho$ is defined as
\begin{equation}
S_\alpha(\rho)=\frac1{1-\alpha}\ln\tr(\rho^\alpha).
\end{equation}
\end{definition}

The von Neumann entropy can be formally written as $S_1$ since
\begin{equation}
S(\rho)=\lim_{\alpha\to1}S_\alpha(\rho).
\end{equation}

\begin{corollary} \label{c:r}
For any constant $\alpha>0$, Theorems \ref{thm:main} and \ref{t:a} remain valid upon replacing $S(\Psi_{j,A})$ and $S(\psi_{j,A})$ in Eqs.~(\ref{eq:main}) and (\ref{eq:a}) by $S_\alpha(\Psi_{j,A})$ and $S_\alpha(\psi_{j,A})$, respectively.
\end{corollary}

\begin{IEEEproof}
Since $S_\alpha$ is monotonically non-increasing in $\alpha$, it suffices to prove the case where $0<\alpha<1$. We follow the proof of Theorem \ref{t:a}. Similar to (\ref{eq:concon}), using the concavity of the R\'enyi entropy for $0<\alpha<1$ \cite{Ras11} and Eq.~(1.42) of Ref.~\cite{Weh78},
\begin{equation}
\frac1{|J|}\sum_{j=1}^{|J|}S_\alpha(\psi_{j,A})\le S_\alpha(\rho_A)\le\frac1{1-\alpha}\ln\sum_{k=1}^{d_A}p_k^\alpha.
\end{equation}
Maximizing the rightmost side subject to the constraint (\ref{eq:cst}),
\begin{align}
&\frac1{1-\alpha}\ln\sum_{k=1}^{d_A}p_k^\alpha\nonumber\\
&\le\frac1{1-\alpha}\ln\big(md_A^{1-\alpha}/2^\alpha+(1-m/2)^\alpha(1-m)^{1-\alpha}d_A^{1-\alpha}\big)\nonumber\\
&=L\ln d+\frac1{1-\alpha}\ln\big(m/2^\alpha+(1-m/2)^\alpha(1-m)^{1-\alpha}\big)\nonumber\\
&=L\ln d-\Theta(1).
\end{align}
\end{IEEEproof}

\begin{IEEEproof} [Proof of Corollary \ref{cor}]
Let $C$ be a sufficiently large constant. Since $C<\Delta=O(\sqrt N)$, Lemma 1 in Ref.~\cite{RCA24} implies that
\begin{equation}
|J|=\Omega(d^N\Delta/\sqrt N).
\end{equation}
Since $|\psi\rangle$ is a Haar-random state in $\mathcal H_J$, Theorem \ref{t:a} implies that
\begin{equation}
\E_{|\psi\rangle\in\mathcal H_J}S(\psi_A)=L\ln d-\Omega(1).
\end{equation}
\begin{lemma} [Lemma III.2 in Ref.~\cite{HLW06}]
For any two pure states $|\phi_1\rangle,|\phi_2\rangle$,
\begin{gather}
|S(\phi_{1,A})-S(\phi_{2,A})|\le\eta\||\phi_1\rangle-|\phi_2\rangle\|,\\
\phi_{i,A}:=\tr_{\bar A}|\phi_i\rangle\langle\phi_i|,\quad i=1,2,
\end{gather}
where the Lipschitz ``constant'' is upper bounded by $\eta=O(\ln d_A)$.
\end{lemma}
Levy's lemma (Lemma III.1 in Ref.~\cite{HLW06}) implies that
\begin{multline}
\Pr_{|\psi\rangle\in\mathcal H_J}\big(S(\psi_A)=L\ln d-\Omega(1)\big)\ge1-e^{-\Omega(|J|/\eta^2)}\\
=1-e^{-\Omega(d^N\Delta/N^{5/2})}.
\end{multline}
\end{IEEEproof}

\section{Lemmas for Section \ref{seth}} \label{s:l}

\begin{lemma} [moments \cite{Ans16, HBZ19, Hua22AP}] \label{l:m}
\begin{gather}
|\tr(H_A^3)|/d_A=O(L),\\
|\tr(H_A^r)|/d_A=(O(Lr))^{r/2},\quad\forall r\in\mathbb N^+,
\end{gather}
where the constant hidden in the Big-O notation is independent of $r$.
\end{lemma}

\begin{lemma} \label{l:et}
For $|\beta|=o(1/\sqrt L)$,
\begin{equation}
\mathcal E(\beta)=-\beta(L-1)s^2/N\pm O(\beta^2L+|\beta|^3L^2).
\end{equation}
\end{lemma}

\begin{IEEEproof}
Using Lemma \ref{l:m},
\begin{align}
&1\le\frac{\tr e^{-\beta H_A}}{d_A}=\sum_{r=0}^\infty\frac{(-\beta)^r\tr(H_A^r)}{r!d_A}\nonumber\\
&\le1+\sum_{r=2}^\infty\frac{|\beta|^r(O(Lr))^{r/2}}{r!}=1+O(\beta^2L),\\
&\left|\sum_{r=2}^\infty\frac{(-\beta)^r\tr(H_A^{r+1})}{r!d_A}\right|=O(\beta^2L)+\sum_{r=3}^\infty\frac{|\beta|^r(O(Lr))^{\frac{r+1}2}}{r!}\nonumber\\
&=O(\beta^2L+|\beta|^3L^2).
\end{align}
Hence,
\begin{align}
\mathcal E(\beta)&=\E_{|A|=L}\tr(\sigma_A(\beta)H_A)=\E_{|A|=L}\sum_{r=0}^\infty\frac{(-\beta)^r\tr(H_A^{r+1})}{r!d_A(1+O(\beta^2L))}\nonumber\\
&=-\frac{\beta\E_{|A|=L}s_A^2}{1+O(\beta^2L)}\pm O(\beta^2L+|\beta|^3L^2)\nonumber\\
&=-\beta(L-1)s^2/N\pm O(\beta^2L+|\beta|^3L^2).
\end{align}
\end{IEEEproof}

\begin{lemma} \label{l:ent}
For $|\beta|=o(1/\sqrt L)$,
\begin{equation}
\mathcal S(\beta)=L\ln d-\frac{\beta^2(L-1)s^2}{2N}\pm O(|\beta|^3L+\beta^4L^2).
\end{equation}
\end{lemma}

\begin{IEEEproof}
It follows from Lemma \ref{l:et} and the thermodynamic relation
\begin{equation}
\mathrm d\mathcal S(\beta)/\mathrm d\beta=\beta\,\mathrm d\mathcal E(\beta)/\mathrm d\beta=\mathrm d(\beta\mathcal E(\beta))/\mathrm d\beta-\mathcal E(\beta).
\end{equation}
\end{IEEEproof}

\section*{Acknowledgment}

I would like to thank Masudul Haque, Ivan M. Khaymovich, and Paul A. McClarty for interesting discussions on related works \cite{Hua21NPB, HMK22} and problems. I also thank \'Alvaro M. Alhambra for explaining some technical details of Lemma 1 in Ref.~\cite{RCA24} and Lemma 8 in Ref.~\cite{BC15}.

\bibliographystyle{IEEEtran}
\bibliography{final}

\end{document}